\documentclass[11pt]{article}
\usepackage[utf8]{inputenc}
\usepackage{amsmath}
\usepackage{amssymb}
\usepackage{amsthm}
\usepackage{url}
\usepackage{hyperref}
\usepackage{xcolor}
\usepackage{dsfont}
\usepackage{float}
\usepackage{mathrsfs}
\usepackage{yfonts}
\newtheorem{definition}{Definition}
\newtheorem{proposition}[definition]{Proposition}

\newtheorem{remark}[definition]{Remark}
\newtheorem{theorem}[definition]{Theorem}
\newtheorem{lemma}[definition]{Lemma}
\newtheorem{corollary}[definition]{Corollary}
\usepackage{graphicx}
\usepackage{fullpage}
 \usepackage{caption}
 \usepackage{subcaption}
 \usepackage{verbatim}
\usepackage[title]{appendix}
\usepackage{upgreek}

\usepackage{tikz}
\usetikzlibrary{shapes.geometric,plotmarks,backgrounds,fit,calc,circuits.ee.IEC}
\usetikzlibrary{patterns}
\usetikzlibrary{decorations.pathreplacing}
\usepackage{pgfplotstable}

\usepackage{epsfig}
\usetikzlibrary{shapes.symbols,patterns} 
\usepackage{pgfplots}
\usepgfplotslibrary{ternary}
\pgfplotsset{compat=newest}

\usepackage{quantikz}

\usepackage{ytableau}

\newcommand{\BC}{\mathbb{C}}
\newcommand{\BN}{\mathbb{N}}

\newcommand{\cC}{{\cal C}}

\newcommand{\cE}{{\cal E}}

\newcommand{\cH}{{\cal H}}

\newcommand{\cT}{{\cal T}}

\usepackage{makecell}
\usepackage{cancel}

\allowdisplaybreaks[4]

\usepackage{braket}

\newcommand{\mr}{\mathrm}

\newcommand{\U}{{\mathrm{U}}}

\newcommand{\BM}{{\mathbb{M}}}
\newcommand{\BU}{{\mathbb{U}}}
\newcommand{\BY}{\mathbb{Y}}

\newcommand{\V}{{\mathscr{V}}}
\newcommand{\VS}{\mathscr{S}}
\newcommand{\VSw}{\mathscr{W}}

\newcommand{\id}{\mathds{1}}
\newcommand{\idm}{\mathrm{id}}
\newcommand{\pr}{\mathbb{P}}

\DeclareMathOperator{\tr}{Tr}

\newcommand{\tp}{{\tt{t}}} 
\newcommand{\ketbra}[2]{|#1\rangle\!\langle #2|}
\let\inner\relax
\NewDocumentCommand\inner{mg}{%
	\ensuremath{\left\langle #1 \middle\vert \IfNoValueTF{#2}{#1}{#2}\right\rangle}%
}

\newcommand{\scp}[3]{\langle #1 \vert #2 \vert #3 \rangle}

\usepackage{listings}

\usepackage{authblk}

\usepackage{soul}

\title{The complexity of semidefinite programs for testing $k$-block-positivity}

\author[1,2,3]{Qian Chen\thanks{chenqian.phys@gmail.com}}
\author[4]{Beno\^it Collins\thanks{collins@math.kyoto-u.ac.jp}}

\affil[1]{Mathematics Research Center, School of Science and Engineering, The Chinese University of Hong Kong, Shenzhen, China}
\affil[2]{Universit\'e de Lyon, Inria, ENS de Lyon, UCBL, LIP, France}
\affil[3]{QICI Quantum Information and Computation Initiative, School of Computing and Data Science, The University of Hong Kong, Pokfulam Road, Hong Kong}
\affil[4]{Kyoto University, Mathematics department, Kitashirakawa Oiwake-cho, 606-8502, Kyoto, Japan}

\date{}

\begin{document}

\maketitle

\begin{abstract}
We extend \cite{chen2025srkbp} by analyzing the complexity of the $k$-block-positivity testing algorithm that stems from the optimization problem in Definition \ref{definition:SDP-k-block-positivity}. In this paper, we investigate a symmetry reduction scheme based on rectangular shaped Young diagrams. Connecting the complexity to the dimensions of irreducible representations of $\U(d)$, we derive an explicit formula for the complexity, which also clarifies why the semidefinite program hierarchy collapses in the $k=d$ case.
\end{abstract}

\section*{Introduction}

A fundamental task in quantum information theory is to characterize linear maps. This involves determining not only whether a map is positive, but also its degree of positivity, that is, for which values of `$k$' it is $k$-positive. The Choi operator of a $k$-positive map is $k$-block-positive, meaning that the operator yields nonnegative expectation value for all pure states with at most Schmidt rank $k$. It closely connects to the notions of bound entanglement, entanglement cost \cite{Hayden:2000zqp,Hayashi:2016gly,Yamasaki:2024nbc} and distillability of entanglement, e.g., the 2-copy distillability conjecture \cite{PRXQuantum.3.010101}. The set of $k$-block-positive operators is the dual of the set consisting of Schmidt number $k$ states \cite{PhysRevA.61.040301,PhysRevA.63.050301,johnston2010NQitI,johnston2010NQitII}.

\medskip

Examining whether a operator is $k$-block-positive or not is a very challenging task.
We extend our work \cite{chen2025srkbp} by analyzing the complexity of the algorithm for testing $k$-block-positivity within the framework of semidefinite programs (SDPs). This algorithm converts the problem from testing $k$-block-positivity to testing $1$-block-positivity via $k$-purification, and then uses extendibility hierarchy to approximate its dual, the set of separable states.

Subsequently, we adopt a scheme that reduces some of the complexity from choosing Young diagrams in \cite{chen2025srkbp}.
The complexity associated with Young diagrams originates from the symmetry reduction under $\U(k)$-symmetry, which defines feasible states as positive semidefinite block matrices indexed by Young diagrams. As the extendibility hierarchy progresses to higher levels, the number of Young diagrams grows, thereby increasing the computational resources required for optimization, since the minimal value must be identified across all blocks. Thus, restricting the set of Young diagrams while maintaining the validity of the test will greatly simplify the algorithm.

Eventually, we study the complexity of this algorithm within the framework of semidefinite programs. Regarding the required SDP variables as computational resources, we derive a formula characterizing the size of this resource, which also provides a quantitative explanation of the phenomenon whereby the extendibility hierarchy should collapse in the $k=d$ case.

\medskip

We will show that establishing hierarchical SDPs on rectangular shape of Young diagrams is sufficient for testing $k$-block-positivity. Then we are led to the main result of this paper.
\begin{theorem}[The SDP complexity of rectangular scheme] \label{thm:main-SDPComplexity}
Setting $N$ by $N+k-1=k n$ with integer $n$. The SDP complexity of the reduced SDP indexed by a rectangular Young diagram $(n^{k})$ is
\begin{align}
\cC_{(n^{k})}
=d \frac{k(d+n-1)}{k+n-1} \prod_{r=1}^{k} \frac{(d+n-r-1)!(k-r)!}{(k+n-r-1)!(d-r)!}.
\label{eq:thm:main-SDPComplexity}
\end{align}
That is to say, the SDP complexity is $O(n^{k(d-k)})$ after symmetry reduction, compared with the unreduced one $O((kd)^{N+1})$.
\end{theorem}
The proof of Theorem~\ref{thm:main-SDPComplexity} guides the structure of this paper as follows.
\begin{itemize}
\item Step 0: Section~\ref{sec:KBP-Optimal} presents the basic notation for $k$-block-positivity, as well as related optimization problems for testing $k$-block-positivity.
\item Step 1: Section~\ref{sec:KBP-SDP} introduces the involved SDPs. In Subsection~\ref{subsec:KBP-SDP-reduced}, we review the reduced SDPs indexed by Young diagrams, and present the questions about (a) economical scheme on Young diagrams and (b) hierarchy collapse if $k=d$. In Subsection~\ref{subsec:RelaxTraceConstraint} we change the trace constraint from equality to inequality, and prepare the construction displayed in Corollary~\ref{cor:mu-PermInvVec} for the sequential analysis.

\item Step 2: Section~\ref{sec:SDP:RectangularYoungDiagrams} is dedicated to answer the question about economical scheme on Young diagrams, showing that testing $k$-block-positivity can rely on hierarchical SDPs indexed by rectangular shape Young diagrams, as stated in Proposition~\ref{pro:SDPs-rectangularYoung}.

\item Step 3: Section~\ref{sec:ComplexitySDP} aims to complete the proof of Theorem~\ref{thm:main-SDPComplexity}. Subsection~\ref{subsec:Complexity-YoungDiagram} provides the relation between the size of SDP variables (indexed by one Young diagram) and the dimension of the irreducible representation of $\U(d)$, as stated in Proposition~\ref{proposition:SDPComplexity-lambda}. We then derive the explicit formula in Theorem~\ref{thm:main-SDPComplexity}, which is followed by two corollaries: one explains why the SDP hierarchy collapses if $k=d$, another shows the equivalent complexity for testing $k$- and $k'$-block positivity if $k+k'=d$.
\end{itemize}



 
\section{Preliminary} \label{sec:KBP-Optimal}

\subsection{Notation}

Let $\mathbb{M}_{m \times n}(\BC)$ denote the space of $m \times n$ matrices over $\BC$, and simply $\BM_{m}(\BC)$ when $m=n$; $\mathrm{Herm}(\mathbb{C}^{d})$ and $\mathrm{Herm}(\mathbb{C}^{d})_{+}$ the subsets of $\mathbb{M}_{d}(\BC)$ of hermitian matrices and positive semidefinite matrices, respectively. Let $\BC^{d_{A}}$ and $\BC^{d_{B}}$ be the original spaces of Alice and Bob, respectively; in this paper, we set $d=d_{A}=d_{B}$ by default.

Let $\BC^{k}$ be the auxiliary spaces carrying $k$-purification (or referred to as $k$-extension). The space after $k$-purification is denoted by $\cH := \BC^{kd} \cong \BC^{k} \otimes \BC^{d}$. We denote the permutation over $\cH^{\otimes N}$ by $\Delta$ which is defined by $\Delta : S_{N} \to \mr{Perm}^{k} \times \mr{Perm}^{d}$ where $\mr{Perm}^{k} : S_{N} \to \U((\BC^{k})^{\otimes N})$ is the natural unitary representation of permutation prescribed by $\mr{Perm}^{k}(\pi) \ket{i_{1} \ldots i_{N}}=\ket{ i_{\pi^{-1}(1)} \ldots i_{\pi^{-1}(N)}}$ for $\pi \in S_{N}$, and likewise for $\mr{Perm}^{d} : S_{N} \to \U((\BC^{d})^{\otimes N})$. That is to say, $\Delta(\pi)$ is defined by permuting auxiliary and original spaces with $\mr{Perm}^{k}(\pi)$ and $\mr{Perm}^{d}(\pi)$ respectively. We will often drop $\mr{Perm}^{k}$ and $\mr{Perm}^{d}$ so write $\Delta(\pi)=\pi \otimes \pi$ when the context is clear.

For Alice's auxiliary system, we introduce linear map $\cE : \BC^{k} \to \mr{Alt}^{k-1} \BC^{k}$ where $\mr{Alt}^{k-1} \BC^{k}$ is the anti-symmetric subrepresentation of $(\BC^{k})^{\otimes (k-1)}$.
Denote by $\ket{\phi_{k}}=\sum_{i=1}^{k} \ket{i i}$ the unnormalized $k$-dimensional maximally entangled state, and define $\hat{\phi}_{k}=\ket{\phi_{k}} \otimes \id_{d_{A}} \otimes \id_{d_{B}}$.
Denote by $\mr{sr}(\psi)$ the Schmidt rank of vector $\ket{\psi}$.

We denote $\lambda \vdash_{k} N$ if $\lambda$ is a Young diagram of $N$ boxes with at most $k$ rows, i.e., $\ell(\lambda) \leq k$. Denote by $\BY_{\lambda}$ the Specht module associated to Young diagram $\lambda$. Let $\lambda^{-}$ denote the Young diagram obtained from $\lambda$ with $\ell(\lambda)=k$ by removing its first column, that is, if $\lambda=(\lambda_1, \ldots , \lambda_k)$, then $\lambda^{-} = (\lambda_1-1, \ldots, \lambda_k-1)$.
Denote by $\BU_{\lambda}^{k}$ and $\BU_{\lambda}^{d}$ the irreducible representations associated with the Young diagram $\lambda$ for the unitary groups $\mathrm{U}(k)$ and $\mathrm{U}(d)$, respectively.

We write $\mu \subseteq \lambda$ for Young diagrams such that $\mu$ is contained in $\lambda$, i.e., $\mu_{i} \leq \lambda_{i}$ for all $i$.
We write $\lambda \searrow \mu$ if Young diagram $\lambda$ is obtained by appending exactly one box to $\mu$.




\subsection{Optimization problems for testing $k$-block positivity} \label{subsec:KBP-Optimization}

This section provides a brief review of $k$-block-positivity and formulates two optimization problems for testing it. Recall that in a bipartite system $\BC^{d_{A}} \otimes \BC^{d_{B}}$, any bipartite pure state with Schmidt rank at most $k$ can be written in the following form:
\begin{align}
\ket{\psi}=\sum_{i=1}^{k} \ket{x_{i}} \otimes \ket{y_{i}}\,, \: \text{ where for all $i$, } \ket{x_{i}} \in \BC^{d_{A}}, \: \ket{y_{i}} \in \BC^{d_{B}}.
\end{align}
Denote by $\mathrm{SR}_{k} (d_{A} , d_{B})$ the set consisting of pure states with Schmidt rank at most $k$.
A hermitian operator $X \in \mathrm{Herm}(\BC^{d_{A}} \otimes \BC^{d_{B}})$ is said to be $k$-block-positive if $X$'s expectation value is nonnegative for all pure states in $\mathrm{SR}_{k} (d_{A} , d_{B})$,
\begin{align}
\scp{\psi}{X}{\psi} \geq 0\,, \: \forall \ket{\psi} \in \mathrm{SR}_{k} (d_{A} , d_{B}).
\end{align}
A mixed state $\rho$ is said to have Schmidt number $k$ if it lies in the convex hull of $\mathrm{SR}_{k} (d_{A} , d_{B})$ \cite{PhysRevA.61.040301,PhysRevA.63.050301}. Denote by $\mathrm{SN}_{k} (d_{A} , d_{B})$ the set of Schmidt number $k$ states. In particular, $\mathrm{SN}_{1} (d_{A} , d_{B})$ is the set of separable states.
Clearly, $X \in \mathrm{Herm}(\BC^{d_{A}} \otimes \BC^{d_{B}})$ is $k$-block-positive iff $\tr (X \rho) \geq 0$ for all $\rho \in \mathrm{SN}_{k} (d_{A} , d_{B})$. Throughout, we assume $d=d_{A}=d_{B}$, while occasionally retaining the subscripts $A$ and $B$ for notational clarity.

\begin{definition}[Optimization: testing $k$-block-positivity] \label{definition:Optimization-k-block-positivity-original}
The $k$-block-positivity of $X \in \mathrm{Herm}(\BC^{d_{A}} \otimes \BC^{d_{B}})$ can be tested by solving the following optimization problem:
\begin{align}
&
\V=
\min \tr{(X \rho)},
\\
&
\text{subject to } \:
\rho \in \mathrm{SN}_{k} (d_{A},d_{B}), \: \text{ and } \tr\rho=1.
\nonumber
\end{align}
$X$ is $k$-block-positive iff $\V \geq 0$.
\end{definition}
In the optimization problem, $\tr(X \rho)$ serves as the \textit{objective function}. The \textit{feasible set} comprises all allowable states $\rho$, which are required to satisfy a set of constraints. This standard terminology is also used in the context of semidefinite programs.

It is clear that the minimal value of this linear objective function over the convex set $\mathrm{SR}_{k} (d_{A},d_{B})$ is attained at an extreme point. Since the extreme points of $\mathrm{SN}_{k} (d_{A},d_{B})$ are contained within the set of pure states of Schmidt rank at most $k$, the minimum can be found by optimizing over pure states in $\mathrm{SR}_{k} (d_{A},d_{B})$. Because of the inclusion $\mathrm{SN}_{1}(d,d) \subset \mathrm{SN}_{2}(d,d) \subset \cdots \subset \mathrm{SN}_{k}(d,d) \subset \cdots \subset \mathrm{SN}_{d-1}(d,d) \subset \mathrm{SN}_{d}(d,d)$, the optimal values of a minimization problem over these sets satisfy the following sequence of inequalities:
\begin{align}
&
\min_{\rho \in \mathrm{SN}_{d}(d,d)} \tr{(X \rho)} \leq \min_{\rho \in \mathrm{SN}_{d-1}(d,d)} \tr{(X \rho)} \leq \cdots \leq \min_{\rho \in \mathrm{SN}_{k}(d,d)} \tr{(X \rho)}
\nonumber \\
&
\leq \cdots \leq \min_{\rho \in \mathrm{SN}_{2}(d,d)} \tr{(X \rho)} \leq \min_{\rho \in \mathrm{SN}_{1}(d,d)} \tr{(X \rho)},
\end{align}
where $\min_{\rho \in \mathrm{SN}_{d}(d,d)} \tr{(X \rho)}$ amounts to solving $X$'s minimal eigenvalue.

\paragraph{Auxiliary spaces with $U^{\otimes k}$-symmetry} We define auxiliary spaces via $k$-purification, and introduce $\cE$ map to convert $\bar{U} \otimes U$-symmetry to $U^{\otimes k}$-symmetry.
\begin{definition}[$k$-purification and dualization]
For any $X \in \mathrm{Herm}(\mathbb{C}^{d} \otimes \mathbb{C}^{d})$ and $\rho \in \mathrm{Herm}(\mathbb{C}^{d} \otimes \mathbb{C}^{d})_{+}$, we define their $k$-purification $X_{k} \in \mathrm{Herm}(\mathbb{C}^{kd} \otimes \mathbb{C}^{kd})$ and $\rho_{k} \in \mathrm{Herm}(\mathbb{C}^{kd} \otimes \mathbb{C}^{kd})_{+}$ by:
\begin{align}
&
X_{k}:=
k \Pi_{k} \otimes X, \:
\text{ where }
\Pi_{k}=\sum_{i_1 , \ldots , i_{k}=1}^{k} \sum_{i'_1 , \ldots , i'_{k}=1}^{k} \frac{\epsilon_{i_1 \ldots i_k}\epsilon_{i'_1 \ldots i'_k}}{k!} \ketbra{i_1 \ldots i_k}{i'_1 \ldots i'_k},
\\
&
\rho_{k}:=
\sum_{i_{0},i_{1},j_{0},j_{1}=1}^{k}
(\cE \otimes \id_{k}) \ketbra{i_{0} i_{1} }{j_{0} j_{1}} (\cE^{\dagger} \otimes \id_{k}) \otimes \rho_{i_{0} i_{1}, j_{0} j_{1}}, \:
\text{ where } \rho_{i_{0} i_{1}, j_{0} j_{1}} \in \mathbb{M}_{d^2 \times d^2}(\mathbb{C}),
\end{align}
where $\Pi_{k}$ is the self-adjoint projector onto the $1$-dimensional fully anti-symmetric space that satisfies $\Pi_{k} \Pi_{k}=\Pi_{k}$, $\Pi_{k}^{\dagger}=\Pi_{k}$ and $\Pi_{k}=U^{\otimes k} \Pi_{k} {U^{\dagger}}^{\otimes k}$ for all $U \in \U(k)$, associated with Young diagram $(1^{k})$; Linear map $\cE : \BC^{k} \to \mr{Alt}^{k-1} \BC^{k}$ induces an isomorphism with the dual space by
\begin{align}
\cE \ket{i}= \sum_{a_2 , \ldots , a_{k}=1}^{k} \frac{\epsilon_{a_2 \ldots a_k i}}{\sqrt{(k-1)!}} \ket{a_2 \ldots a_k}.
\end{align}
\end{definition}
When the domain is clear from the context, we will abuse notation by also writing $\cE$ to denote its trivial extension $\cE \otimes \id$, where $\id$ is the identity map acting on the auxiliary copies of Bob.
$\cE$ satisfies $\cE^{\dagger} \cE=\id_{k}$ and $\cE \cE^{\dagger}=\pr_{(1^{k-1})}$, which is the projector of $\mr{Alt}^{k-1} \BC^{k}$, and
\begin{align}
k \cE^{\dagger} \Pi_{k} \cE
=
\ketbra{\phi_{k}}{\phi_{k}},
\: \text{ where }
\ket{\phi_{k}}=\sum_{i=1}^{k} \ket{i i}.
\label{eq:Dualization-MaximallyEntangledProj}
\end{align}
The greater theoretical simplicity of the Schur-Weyl duality motivates our use of $\cE$ to convert the $\bar{U} \otimes U$-symmetry carried by $\ketbra{\phi_{k}}{\phi_{k}}$ into the $U^{\otimes k}$-symmetry carried by $\Pi_{k}$ \cite{chen2025srkbp}, even though the invariants of $\bar{U} \otimes U$ are not the same as those of $U^{\otimes k}$ in principle.

\paragraph{Representing feasible pure states}
\begin{proposition} \label{Proposition:PureStateKraus-xy}
Every $\ket{\varphi} \in \mr{SR}_{1} (kd_{A}, kd_{B})$ can be represented in the following manner by some $x \in \BM_{d_{A} \times k}(\BC)$ and $y \in \BM_{d_{B} \times k}(\BC)$ that satisfy normalization $\tr(x^{\dagger} x)\tr(y^{\dagger} y)=1$:
\begin{align}
\mr{SR}_{1} (kd_{A},kd_{B}) \ni \ket{\varphi}=\sum_{i,j=1}^{k} \ket{i j} \otimes (x \otimes y) \ket{i j}.
\label{eq:PureStateKraus-xy}
\end{align}
Since $\cE$ is a local operation on Alice's system, $\cE \ket{\varphi}$ is also separable with respect to the bipartition between Alice and Bob, hence we write $\cE \ket{\varphi} \in \mr{SR}_{1} (kd_{A},kd_{B})$ despite a little abuse notation.
\end{proposition}
\begin{proof}
Write $\ket{\varphi}=\ket{\varphi_{A} \otimes \varphi_{B}}$ with $\ket{\varphi_{A}}=\sum_{i=1}^{k} \sum_{r=1}^{d_{A}} x_{ir}\ket{ir}=\ket{i} \otimes x \ket{i}$ and $\ket{\varphi_{B}}$ likewise.
\end{proof}

\paragraph{Optimization problem based on $k$-purification}
\begin{definition}[Optimization: testing $k$-block-positivity via $k$-purification] \label{definition:SDP-k-block-positivity}
The $k$-block-positivity of $X \in \mathrm{Herm}(\BC^{d_{A}} \otimes \BC^{d_{B}})$ can be tested by solving the following optimization problem:
\begin{align}
&
\V_{k}=
\min \tr{(X_{k} \rho_{k})},
\\
&
\text{subject to } \:
\rho_{k} \in \mr{SN}_{1} (kd,kd), \: \text{ and } \tr\rho_{k}=1.
\nonumber
\end{align}
Because $X$ is $k$-block positive iff $X_{k}$ is block positive, and $X_{k}$ is block-positive iff $\V_{k}=0$\,.
\end{definition}
Notably, the criterion changes from $\V \geq 0$ to $\V_{k}=0$, for the reason that the $1$-dimensional projector introduced by $k$-purification has nontrivial kernel. For instance, if $x y^{\tp}=0$ then $x \otimes y \ket{\phi_{k}}=0$, thus $\hat{\phi}_{k}^{\dagger} \ket{\varphi}=0$ where $\hat{\phi}_{k}^{\dagger}=\bra{\phi_{k}} \otimes \id_{d_{A}} \otimes \id_{d_{B}}$. For more details, see the following proposition, whose proof is relegated to Appendix~\ref{proof:proposition:Bound-SNK-SN1-SNk}.
\begin{proposition} \label{proposition:Bound-SNK-SN1-SNk}
The following relation holds for the minimal values of optimization problems defined in Definition~\ref{definition:Optimization-k-block-positivity-original} and Definition~\ref{definition:SDP-k-block-positivity}:
\begin{itemize}
\item If $\V \geq 0$, then $\V_{k}=0$.
\item Otherwise, $k \V_{k} \leq \V \leq \V_{k} < 0$.
\end{itemize}
\end{proposition}




\section{Semidefinite programs for testing $k$-block positivity} \label{sec:KBP-SDP}

\subsection{Reduced semidefinite programs indexed by Young diagrams} \label{subsec:KBP-SDP-reduced}

The optimization problem defined in Definition~\ref{definition:SDP-k-block-positivity} is difficult to solve in principle,
but the hierarchical SDPs constructed from extendibility furnishes increasingly tight approximations, which are guaranteed to converge to the true solution as the levels of hierarchy tend to infinity.
\begin{definition}[$k$-block-positivity testing SDP with $N$-Bose symmetric extension] \label{definition:SDP-k-block-positivity-N-BSE}
The $N$-level SDP of the Bosonic extendibility hierarchy is defined as follows,
\begin{align}
&
\VS_{N}:=
\min \tr{(X_{k,N} \rho_{k,N})}, \\
&
\text{subject to } \:
\rho_{k,N} \in \mathrm{Bos}(\mathcal{H}_{A} , \mathcal{H}_{B}^{\otimes N})_{+}, \ \text{ and } \tr\rho_{k, N}=1,
\nonumber
\end{align}
where $\cH=\BC^{k} \otimes \BC^{d}$. Here the $X_{k,N}$ and $\rho_{k,N}$ are defined to be
\begin{align}
X_{k,N}
&=
k \Pi_{k} \otimes \id_{k}^{\otimes (N-1)} \otimes X \otimes \id_{d_{B}}^{\otimes (N-1)},
\label{eq:X-kExtension-NExtension}
\\
\rho_{k,N}
&=
\sum_{i_{0} , \ldots , j_{N}=1}^{k}
\cE \ketbra{i_{0} i_{1} \cdots i_{N}}{j_{0} j_{1} \cdots j_{N}} \cE^{\dagger}
\otimes
\rho_{i_{0} i_{1} \cdots i_{N}, j_{0} j_{1} \cdots j_{N}},
\label{eq:rho-kExtension-NExtension}
\end{align}
where $\rho_{i_{0} i_{1} \cdots i_{N}, j_{0} j_{1} \cdots j_{N}} \in \mathbb{M}_{d_{A}d_{B}^{N}}(\mathbb{C})$.
The constraint $\rho_{k,N} \in \mathrm{Bos}(\mathcal{H}_{A} , \mathcal{H}_{B}^{\otimes N})_{+}$ requires that $\rho_{k,N} \geq 0$ and that $\rho_{k,N}$ is invariant under all permutations of Bob’s systems, i.e., $\rho_{k,N}=\Delta(\pi) \rho_{k,N}=\rho_{k,N} \Delta(\pi)$ for all $\pi \in S_{N}$, the so-called Bosonic permutation symmetry, where the left and right actions of $\Delta(\pi)$ are defined as follows:
\begin{align}
\begin{aligned}
\Delta(\pi) \rho_{k,N}
&:=
\sum_{i_{0} , \ldots , j_{N}=1}^{k}
(\cE \otimes \pi) \ketbra{i_{0} i_{1} \cdots i_{N}}{j_{0} j_{1} \cdots j_{N}} \cE^{\dagger}
\otimes
(\id_{d_{A}} \otimes \pi) \rho_{i_{0} i_{1} \cdots i_{N}, j_{0} j_{1} \cdots j_{N}},
\\
\rho_{k,N} \Delta(\pi)
&:=
\sum_{i_{0} , \ldots , j_{N}=1}^{k}
\cE \ketbra{i_{0} i_{1} \cdots i_{N}}{j_{0} j_{1} \cdots j_{N}} (\cE^{\dagger} \otimes \pi)
\otimes
\rho_{i_{0} i_{1} \cdots i_{N}, j_{0} j_{1} \cdots j_{N}} (\id_{d_{A}} \otimes \pi).
\end{aligned}
\label{eq:Permutation-kExtension}
\end{align}
\end{definition}
An alternative approach formulates the hierarchical SDPs via $N$-extendible states, i.e., states that arise as the marginal of a state that is symmetric under permutations of $N$ subsystems. In our formulation, the trivial extension $\id_{k}^{\otimes (N-1)} \otimes \id_{d_{B}}^{\otimes (N-1)}$ is used to construct such symmetric extensions.

It is clear that the reduced state of a Bose symmetric state is Bose symmetric as well. Consequently, $\VS_{N} \leq \VS_{N+1}$ holds for every $N$, and eventually in the limit of infinite level $\VS_{\infty} := \lim_{N \to \infty} \VS_{N}=\V_{k}$ due to quantum de Finetti theorem~\cite{caves2002DeFinetti,christandl2007DeFinetti}. Therefore,
\begin{align}
\VS_{1} \leq \cdots \leq \VS_{N} \leq \cdots \leq \VS_{\infty} = \V_{k} \leq 0.
\label{eq:sequence-SDPminimalvalue}
\end{align}

\medskip

One can use the symmetries for reducing the difficulty of solving SDP.
Formally, suppose a group $G$ acts linearly on the given element $\mathcal{X} \in V$ via a map $\mathfrak{O} : G \to \mathrm{End}(V)$ and assume $\mathfrak{O}_{g}(\mathcal{X})=\mathcal{X}$ for all $g \in G$. Then on the other hand, the feasible states are invariant under the $\mathfrak{O}^{*}$: the equality $\langle\rho,\mathcal{X}\rangle=\langle\rho,\mathfrak{O}_{g}(\mathcal{X})\rangle=\langle\mathfrak{O}^{*}(\rho),\mathcal{X}\rangle$ for every $\rho \in \mathscr{F}$ implies that we can restrict on $\mathfrak{O}^{*}(\mathscr{F})$ where the $\langle\rho,\mathcal{X}\rangle$ is the objective function and $\mathscr{F}$ the set of feasible states. Therefore, by decomposing the representation $\mathfrak{O}$ into irreducible components, we obtain $\mathcal{X} \cong \bigoplus_{\lambda} \mathcal{X}_{\lambda}$ and $\mathfrak{O}^{*}(\rho) \cong \bigoplus_{\lambda} \rho_{\lambda}$ where $\lambda$ labels irreducible representations of $G$, and each $X_{\lambda}$ and $\rho_{\lambda}$ is a smaller component. Hence, the original SDP yields a set of smaller independent SDPs, one per $\lambda$, which reduces the problem size.

The SDP in Definition~\ref{definition:SDP-k-block-positivity-N-BSE} admits symmetry reductions.
Roughly speaking, the space to be considered, admits: (i) auxiliary unitary $\U(k)$-symmetry carried by $(\BC^{k})^{\otimes (N+k-1)}$ from $k$-purified space $(\BC^{k})^{\otimes (N+k-1)} \otimes (\BC^{d})^{\otimes (N+1)}$ where Alice's $(\BC^{k})^{\otimes (k-1)}$ arises from the dualization $\cE : \BC^{k} \to \mr{Alt}^{k-1} \BC^{k} \subset (\BC^{k})^{\otimes (k-1)}$; (ii) Bosonic permutation symmetry from extendibility carried by Bob's copies $(\BC^{kd})^{\otimes N}$. To be more explicit, the two symmetries are listed as follows:
\begin{itemize}
\item the $\U(k)$-symmetry carried by $X_{k,N}$,
\begin{align}
X_{k,N}=( U^{\otimes (N+k-1)} \otimes \id_{d_{A}} \otimes \id_{d_{B}}^{\otimes N} ) X_{k,N} ( {U^{\dagger}}^{\otimes (N+k-1)} \otimes \id_{d_{A}} \otimes \id_{d_{B}}^{\otimes N} )
, \: \ \forall U \in \U(k);
\end{align}

\item the Bosonic permutation symmetry carried by $\rho_{k,N}$,
\begin{align}
\rho_{k,N}=\Delta(\pi) \rho_{k,N}=\rho_{k,N} \Delta(\pi)
, \: \ \forall \pi \in S_{N}.
\end{align}
\end{itemize}

\paragraph{SDP symmetry reduction} The $\U(k)$-symmetry allows the $N$-level SDP to be decomposed into independent SDPs indexed by Young diagrams, leading to reduced SDPs~\cite{chen2025srkbp},
\begin{align}
X_{k,N}
&=
\bigoplus_{\substack{\lambda \vdash (N+k-1), \\ \ell(\lambda)=k }}
X_{\lambda},
\: \text{ where } X_{\lambda}=k \Pi_{k} \otimes \pr_{\lambda^{-}} \otimes X \otimes \id_{d}^{\otimes (N-1)},
\label{eq:Uk-invariant-X}
\\
\cT(\rho_{k,N})
&=\int_{\mr{U}(k)} ( {U}^{\otimes (N+k-1)} \otimes \id_{d_{A}} \otimes \id_{d_{B}}^{\otimes N} ) \rho_{k,N} ( {U^{\dagger}}^{\otimes (N+k-1)} \otimes \id_{d_{A}} \otimes \id_{d_{B}}^{\otimes N} ) dU
\nonumber \\
&\cong \bigoplus_{\lambda \vdash_{k} (N+k-1)} w_{\lambda}
\rho_{\lambda}
, \: \text{ where }
\rho_{\lambda}=\frac{\id_{\BU^{k}_{\lambda}}}{\dim \BU^{k}_{\lambda}} \otimes \sum_{p_{\lambda} , p'_{\lambda}}
\cE \cE^{\dagger}
\ketbra{p_{\lambda}}{p'_{\lambda}}
\cE \cE^{\dagger}
\otimes \rho_{p_{\lambda} p'_{\lambda}},
\label{eq:Uk-invariant-rho}
\end{align}
where $\cT(\rho_{k,N})$ is the $\U(k)$-invariant state on $\rho_{k,N}$, satisfying $\tr{(X_{k,N} \cT(\rho_{k,N}))}=\tr{(X_{k,N} \rho_{k,N})}$.
The notations are explained as follows:
\begin{enumerate}
\item $w_{\lambda} \geq 0$ satisfies $\sum_{\lambda} w_{\lambda}=1$, and $\rho_{\lambda}$ satisfies $\tr \rho_{\lambda}=1$ and $\rho_{\lambda}=\Delta(\pi) \rho_{\lambda}=\rho_{\lambda} \Delta(\pi)$ for all $\pi \in S_{N}$.

\item The symbols $p_{\lambda}$ and $p'_{\lambda}$ are standard Young tableaux of shape $\lambda$, and $\rho_{p_{\lambda} p'_{\lambda}} \in \BM_{d_{A} d_{B}^{N}}(\BC)$.

\item Young diagram $\lambda^{-}$ is related to $\lambda=(\lambda_1,\ldots,\lambda_k)$ by $\lambda^{-}=(\lambda_1-1,\ldots,\lambda_k-1)$, or say, $\lambda$ is obtained by appending $(1^{k})$ to $\lambda^{-}$. Applying Littlewood-Richardson rule, we have
\begin{align}
\Pi_{k} \otimes \pr_{\lambda^{-}}=\Pi_{k} \pr_{\lambda} \Pi_{k}
+\text{other Young projectors with rows more than $k$},
\nonumber
\end{align}
while the other Young projectors will annihilate the elements of $(\BC^{k})^{\otimes (N+k-1)}$, thus from now on we look
\begin{align}
\Pi_{k} \otimes \pr_{\lambda^{-}}=\Pi_{k} \pr_{\lambda} \Pi_{k}
\label{eq:ProjectionEquality-klambda-}
\end{align}
where $\pr_{\lambda^{-}}$ and $\pr_{\lambda}$ are the central projections corresponding to Young diagrams $\lambda^{-}$ and $\lambda$ respectively.
Without the restriction $\ell(\lambda) \leq k$, the recoupling process could yield Young diagrams with more than $k$ rows. The decomposition of the auxiliary part of $\rho_{k,N}$ is constrained to Young diagrams with at most $k$ rows, which we denote by $\lambda \vdash_{k} (N+k-1)$, rather than the full set $\lambda \vdash (N+k-1)$.

\end{enumerate}

\begin{definition}[Reduced SDP: $\U(k)$-symmetry] \label{definition:ReducedSDP-Uk}
Let $\lambda \vdash_{k} (N+k-1)$, we define
\begin{align}
&
\VS_{\lambda}=\min_{\rho_{\lambda} \geq 0} \tr (X_{\lambda} \rho_{\lambda}), \\
&
\text{ subject to } \rho_{\lambda}=\Delta(\pi) \rho_{\lambda}=\rho_{\lambda} \Delta(\pi), \: \forall \pi \in S_{N},
\: \text{ and }
\tr \rho_{\lambda} = 1,
\end{align}
where $\rho_{\lambda}$, given by Eq.\eqref{eq:Uk-invariant-rho}, is supported on the subspace corresponding to the block indexed by $\lambda$,
\begin{align}
\rho_{\lambda}
=
(\pr_{\lambda} \otimes \id_{d_{A}} \otimes \id_{d_{B}}^{\otimes N} ) \rho_{\lambda} (\pr_{\lambda} \otimes \id_{d_{A}} \otimes \id_{d_{B}}^{\otimes N} ).
\label{eq:ReducedSDP-Uk-LambdaConfinement}
\end{align}
Here $\pr_{\lambda}$ is the central projection of $\lambda$.
Then the SDP in Definition~\ref{definition:SDP-k-block-positivity-N-BSE} is solved by (see \cite{chen2025srkbp})
\begin{align}
&
\VS_{N}=
\min_{\lambda_{k} \vdash (N+k-1)} \VS_{\lambda}.
\label{eq:VSN-minVSlambda}
\end{align}
\end{definition}

The Young diagram $\lambda$ that appears in Eq.\eqref{eq:ProjectionEquality-klambda-} has exactly $k$ rows. Therefore, any block indexed by a Young diagram with fewer than $k$ rows lies in the kernel of $X_{k,N}$ and contributes zero to the objective function. Consequently, only blocks indexed by Young diagrams with exactly $k$ rows need to be considered.

\paragraph{Two questions}
This paper aims to address the following questions:
\begin{itemize}
\item Economical scheme on Young diagrams: Solving all $\VS_{\lambda}$ requires substantial computational resources, especially when $N$ is large. Is it possible to avoid running over all Young diagrams, meanwhile still ensuring that the $k$-block-positivity test remains valid?
\item Hierarchy collapse: How can we characterize the computational resources or the complexity of the algorithm? In the case of $k=d$, testing $d$-block-positivity should be straightforward, and the extendibility hierarchy is unnecessary, since it only depends on the minimal eigenvalue of $X$. Can the hierarchy collapse be read from the characterization?
\end{itemize}

Our response to the second question will build upon the answer to the first. To lay the groundwork for a clear analysis of the first question, the next subsection is devoted to a tiny revision of the SDP in Definition~\ref{definition:ReducedSDP-Uk}.




\subsection{Relaxation of the trace constraint} \label{subsec:RelaxTraceConstraint}

Given a permutation symmetric vector, we can generate permutation symmetric vectors via central elements of $\BC[S_{N}]$.
\begin{proposition}
Suppose that the vector $\ket{\psi}$ carries permutation symmetry, namely, $\ket{\psi}=\Delta(\pi) \ket{\psi}$ for any $\pi \in S_{N}$. Let $z$ be a central element of $\BC[S_{N}]$, i.e., $z \pi=\pi z$ for all $\pi \in S_{N}$, then the vector $(z \otimes \idm) \ket{\psi}$ also carries permutation symmetry, where we write shorthand $\idm=\id_{d_{A}} \otimes \id_{d_{B}}^{\otimes N}$.
\end{proposition}
\begin{proof}
It is straightforward to check for any $\pi \in S_{N}$
\begin{align}
\Delta(\pi) (z \otimes \idm) \ket{\psi}
=
(\pi \otimes \pi) (z \otimes \idm) \ket{\psi}
=
(z \pi \otimes \pi) \ket{\psi}
=
(z \otimes \idm) \ket{\psi}.
\nonumber
\end{align}
\end{proof}
\paragraph{An important construction}
The above proposition yields a corollary which enables an important construction for the subsequent analysis.
\begin{corollary} \label{cor:mu-PermInvVec}
Let $\mu \vdash N$ with $\ell(\mu) \leq k$
and $\mu \searrow \lambda^{-}$,
the following vector carries permutation symmetry,
\begin{align}
\ket{\varphi_{\lambda^{-}}(x,y)}
=
\sum_{i_{0},i_{1},\ldots,i_{N}=1}^{k}
\sum_{\mu : \mu \searrow \lambda^{-}}
(\cE \otimes P_{\mu}) \ket{i_0 i_1 \ldots i_N} \otimes (x \otimes y^{\otimes N}) \ket{i_0 i_1 \ldots i_N},
\end{align}
where $\pr_{\mu}$ is the central projection of $\mu$. Linear maps $x : \BC^{k} \to \BC^{d_{A}}$ and $y : \BC^{k} \to \BC^{d_{B}}$ can be represented as matrices $x \in \BM_{d_{A} \times k}(\BC)$ and $y \in \BM_{d_{B} \times k}(\BC)$, respectively.
Obviously, $\ket{\varphi_{\lambda^{-}}(x,y)}=\Delta(\pi) \ket{\varphi_{\lambda^{-}}(x,y)}$ for any $\pi \in S_{N}$.
\end{corollary}

\paragraph{Weak trace constraint}
Let us slightly relax the trace constraint in SDP defined in Definition~\ref{definition:ReducedSDP-Uk}.
\begin{definition}[Reduced SDP: free blocks] \label{definition:ReducedSDP-revised}
For a $\lambda \vdash (N+k-1)$ with $\ell(\lambda) \leq k$, we define
\begin{align}
&
\VSw_{\lambda}=\min_{\sigma_{\lambda} \geq 0} \tr (X_{\lambda} \sigma_{\lambda}), \\
&
\text{ subject to } \sigma_{\lambda}=\Delta(\pi) \sigma_{\lambda}=\sigma_{\lambda} \Delta(\pi), \: \forall \pi \in S_{N},
\: \text{ and }
\tr \sigma_{\lambda} \leq 1.
\end{align}
\end{definition}

\begin{proposition}
Comparing the SDPs defined in Definition~\ref{definition:ReducedSDP-Uk} with the one of Definition~\ref{definition:ReducedSDP-revised}, we have (i) $\VS_{\lambda} \geq 0$ iff $\VSw_{\lambda} = 0$ and (ii) $\VS_{\lambda} \leq 0$ iff $\VSw_{\lambda} = \VS_{\lambda}$.
\end{proposition}
\begin{proof}
We start by showing $\VS_{\lambda} \geq 0 \implies \VSw_{\lambda} = 0$ and $\VS_{\lambda} \leq 0 \implies \VSw_{\lambda} = \VS_{\lambda}$.
The inequality $\VSw_{\lambda} \leq \VS_{\lambda}$ is obvious. 
The element $\Pi_{k} \otimes \pr_{\lambda^{-}}$ for $X_{\lambda}$ in Eq.\eqref{eq:Uk-invariant-X} is invariant under the central projection of $\lambda$,
\begin{align}
\Pi_{k} \otimes \pr_{\lambda^{-}}=\pr_{\lambda} ( \Pi_{k} \otimes \pr_{\lambda^{-}} )=( \Pi_{k} \otimes \pr_{\lambda^{-}} ) \pr_{\lambda}.
\end{align}
Let $\sigma \geq 0$ be a Bosonic permutation symmetric state, namely, $\sigma=\Delta(\pi) \sigma=\sigma \Delta(\pi)$ for all $\pi \in S_{N}$ and $\tr \sigma=1$. Note that the following equality holds,
\begin{align}
\tr{(X_{\lambda} \sigma)}
=
\tr{( (\pr_{\lambda} \otimes \idm ) X_{\lambda} (\pr_{\lambda} \otimes \idm ) \sigma)}
=
\tr{( X_{\lambda} (\pr_{\lambda} \otimes \idm ) \sigma (\pr_{\lambda} \otimes \idm ))},
\end{align}
which defines a $\sigma_{\lambda}$ by $\sigma_{\lambda}=(\pr_{\lambda} \otimes \idm) \sigma (\pr_{\lambda} \otimes \idm)$ where $\tr \sigma_{\lambda} \leq 1$ holds due to H\"older inequality
\begin{align}
\tr \sigma_{\lambda}=\tr{( (\pr_{\lambda} \otimes \idm ) \sigma (\pr_{\lambda} \otimes \idm ) )}
\leq
\| \pr_{\lambda} \otimes \idm \|_{\infty} \| \sigma \|_{1}
=1.
\end{align}

On the one hand, since $\pr_{\lambda}$ has a nonzero kernel, it is clear that if $\VS_{\lambda} > 0$, there exists $\sigma$ such that $\sigma_{\lambda}=(\pr_{\lambda} \otimes \idm) \sigma (\pr_{\lambda} \otimes \idm)=0$ then $\tr{(X_{\lambda} \sigma_{\lambda})}=0$. On the other hand, if $\VS_{\lambda} \leq 0$, let us assume there exists some $\sigma_{\lambda}$ yielding $\tr{(X_{\lambda} \sigma_{\lambda})} < \VS_{\lambda}$. Then that $\sigma_{\lambda}$ can be renormalized by $\sigma_{\lambda} \mapsto \rho_{\lambda}=\sigma_{\lambda} / \tr \sigma_{\lambda}$ which yields a smaller value for the SDP in Definition~\ref{definition:ReducedSDP-Uk}, contradicting the assumption $\VSw_{\lambda} < \VS_{\lambda}$. Therefore, $\VS_{\lambda} \leq 0$ implies $\VS_{\lambda}=\VSw_{\lambda}$. The converse follows by the same argument.
\end{proof}




\section{Hierarchical SDPs on rectangular Young diagrams} \label{sec:SDP:RectangularYoungDiagrams}

This section aims to investigate the possibility of restricting the type of Young diagrams on which it remains sufficient to solve our SDP.




\subsection{Tensor power representation and Schur function}

This subsection follows the construction in Corollary~\ref{cor:mu-PermInvVec} and then provides the relation between the objective functions of the SDP defined in Definition~\ref{definition:ReducedSDP-revised} and of the optimization problem defined in Definition~\ref{definition:SDP-k-block-positivity}.

Since the feasible sets of our optimization and SDPs are convex, we can restrict our attention to pure states and the objective functions they yield.
\begin{proposition} \label{proposition:ObjectiveFunction-SchurFunction}
Let $\lambda \vdash_{k} (N+k-1)$, with $\ell(\lambda)=k$, specified by $\lambda=(\lambda_1,\ldots,\lambda_k)$, and $\lambda^{-}=(\lambda_1-1,\ldots,\lambda_k-1)$, and $\mu \vdash_{k} N$ satisfying $\mu \subset \lambda$ and $\mu \searrow \lambda^{-}$.
We consider a nonzero $\ket{\varphi_{\lambda^{-}}(x,y)} \in \BC^{kd_{A}} \otimes \bigoplus_{\mu : \mu \searrow \lambda^{-}} \mr{Sym}^{\mu} (kd_{B})$ defined in Corollary~\ref{cor:mu-PermInvVec}, which yields the objective function of the SDP in Definition~\ref{definition:ReducedSDP-revised}:
\begin{align}
&
f_{\lambda}(x,y)
:=
\frac{\scp{\varphi_{\lambda^{-}}(x,y)}{X_{\lambda}}{\varphi_{\lambda^{-}}(x,y)}}{\inner{\varphi_{\lambda^{-}}(x,y)}{\varphi_{\lambda^{-}}(x,y)}}
=
f(x,y)
\frac{\dim \BY_{\lambda^{-}} s_{\lambda^{-}} (\eta_1 , \ldots , \eta_k) \tr (y^{\dagger} y)
}{\sum_{\mu : \mu \searrow \lambda^{-}} \dim \BY_{\mu} s_{\mu} (\eta_1 , \ldots , \eta_k)},
\label{eq:ObjectiveFunction-SchurFunction}
\\
&
\text{where }
f(x,y)
:=
\frac{\scp{\phi_{k}}{(x \otimes y)^{\dagger} X (x\otimes y)}{\phi_{k}}}{\tr(x^{\dagger} x)\tr(y^{\dagger} y)},
\label{eq:ObjectiveFunction-kExtension:f}
\end{align}
where $s_{\lambda^{-}}$ and $s_{\mu}$ are the Schur functions associated with $\lambda^{-}$ and $\mu$ respectively, whose arguments $\eta_{1} , \ldots , \eta_{k}$ are eigenvalues of $\eta=y^{\dagger} y$, i.e., $\eta=U^{\dagger} \mr{diag}(\eta_1 , \ldots , \eta_k) U$ for some $U \in \U(k)$.

Then, the function $f(x,y)$ is the objective function of the optimization problem defined in Definition~\ref{definition:SDP-k-block-positivity} on pure states defined in Eq.\eqref{eq:PureStateKraus-xy} in Proposition~\ref{Proposition:PureStateKraus-xy} under the setting $\tilde{x} := x / \sqrt{\tr (x^{\dagger} x)}$ and $\tilde{y} := y / \sqrt{\tr (y^{\dagger} y)}$ with $\tr(\tilde{x}^{\dagger} \tilde{x})\tr(\tilde{y}^{\dagger} \tilde{y})=1$.
\end{proposition}
\begin{proof}
We begin with computing $\inner{\varphi_{\lambda^{-}}(x,y)}{\varphi_{\lambda^{-}}(x,y)}$ obtains
\begin{align}
\inner{\varphi_{\lambda^{-}}(x,y)}{\varphi_{\lambda^{-}}(x,y)}
&=
\sum_{\mu : \mu \searrow \lambda^{-}}
\tr(x^{\dagger} x) \scp{j_{1} \ldots j_{N}}{\pr_{\mu}}{i_{1} \ldots i_{N}} \scp{j_{1} \ldots j_{N}}{ (y^{\dagger} y)^{\otimes N}}{i_{1} \ldots i_{N}},
\\
&=
\tr(x^{\dagger} x)
\sum_{\mu : \mu \searrow \lambda^{-}}
\tr( \pr_{\mu} \eta^{\otimes N} )
\\
&=
\tr(x^{\dagger} x)
\sum_{\mu : \mu \searrow \lambda^{-}}
\dim \BY_{\mu} s_{\mu} (\eta_1 , \ldots , \eta_k).
\end{align}
The calculation for the $\tr ( \pr_{\mu} \eta^{\otimes N} )$ relies on the following expression of the central projector $\pr_{\mu}$,
\begin{align}
\pr_{\mu}=\frac{\dim \BY_{\mu}}{N!} \sum_{\pi \in S_{N}} \chi_{\mu}(\pi) \pi,
\end{align}
and then the formula
\begin{align}
\tr ( \pi \eta^{\otimes N} )=\tr ( \pi (U^{\dagger} \mr{diag}(\eta_1 , \ldots , \eta_k) U)^{\otimes N} )
=
\tr ( \pi \mr{diag}(\eta_1 , \ldots , \eta_k)^{\otimes N} )
=
p_{c(\pi)}(\eta_1 , \ldots , \eta_k),
\end{align}
where $p_{c(\pi)}$ is the power sum
\footnote{Given an integer $l \geq 1$, power sum $p_{l}$ is defined to be $p(x_1,\ldots,x_n)=\sum_{i=1}^{l} x_{i}^{l}$. Let $S_{N} \ni \pi=\pi_1 \cdots \pi_m$ then $p_{c(\pi)}:=p_{| \pi_1 |} \cdots p_{| \pi_m |}$ where $| \sigma |$ is the length of permutation $\sigma \in S_{N}$.}
for cycle type $c(\pi)$ of $\pi \in S_{N}$. Then using Frobenius' formula \cite{Macdonald1998SymmetricFun}
\begin{align}
s_{\mu}=\sum_{c(\pi) \in \mathrm{Par}(N)} \frac{1}{z_{c(\pi)}} \chi_{\mu}( c(\pi) ) p_{c(\pi)}, \:
\text{ where }
z_{c(\pi)}=\frac{n!}{1^{m_{1}}m_{1}! 2^{m_{2}}m_{2}! \cdots n^{m_{n}} m_{n}! }
\end{align}
we are able to go from the power sum to Schur function, where $z_{c(\pi)}$ is the cardinality of the conjugacy class to which $c(\pi)$ belongs (Macdonald’s notation).

Then we compute $\scp{\varphi_{\lambda^{-}}(x,y)}{X_{\lambda}}{\varphi_{\lambda^{-}}(x,y)}$. Note that $(\cE^{\dagger} \otimes \sum_{\mu : \mu \searrow \lambda^{-}} \pr_{\mu}) \cdot (k\Pi_{k} \otimes \pr_{\lambda^{-}}) \cdot (\cE \otimes \sum_{\mu' : \mu' \searrow \lambda^{-}} \pr_{\mu'})=\ketbra{\phi_{k}}{\phi_{k}} \otimes \pr_{\lambda^{-}}$ since the following equality holds \footnote{This can be shown by using the nested decomposition of Hermitian Young projectors \cite{Alcock-Zeilinger:2016sxc}.}:
\begin{align}
(\pr_{(1^{k-1})} \otimes \sum_{\mu : \mu \searrow \lambda^{-}} \pr_{\mu}) \cdot (k\Pi_{k} \otimes \pr_{\lambda^{-}}) \cdot (\pr_{(1^{k-1})} \otimes \sum_{\mu' : \mu' \searrow \lambda^{-}} \pr_{\mu'})
=
k \Pi_{k} \otimes \pr_{\lambda^{-}}.
\end{align}
Then by a similar calculation for $\inner{\varphi_{\lambda^{-}}(x,y)}{\varphi_{\lambda^{-}}(x,y)}$, we obtain $\scp{\varphi_{\lambda^{-}}(x,y)}{X_{\lambda}}{\varphi_{\lambda^{-}}(x,y)}$ by using $\tr (\pr_{\lambda^{-}} \eta^{\otimes (N-1)})=\dim \BY_{\lambda^{-}} s_{\lambda^{-}} (\eta_1 , \ldots , \eta_k)$. The proof follows from combining these results.
\end{proof}




\subsection{Choice of Young diagrams: Rectangular shape scheme} \label{subsec:RectangleYoungDiagrams-Schur}

This subsection addresses the question of economical scheme. By establishing the hierarchical SDPs on the rectangular shape Young diagram and using the relation Eq.\eqref{eq:ObjectiveFunction-SchurFunction} in Proposition~\ref{proposition:ObjectiveFunction-SchurFunction}, we can derive a bound Eq.\eqref{eq:Bound-rectangularYoung} and then proceed with the analysis of the validity on the rectangular shape scheme.

Recall that the sequence of inequalities Eq.\eqref{eq:sequence-SDPminimalvalue} holds for the hierarchical SDPs in Definition~\ref{definition:SDP-k-block-positivity-N-BSE}. We can extract the levels $N=kn-k+1$ from the sequence Eq.\eqref{eq:sequence-SDPminimalvalue} such that
\begin{align}
\VS_{1} \leq \VS_{k+1} \leq \cdots \leq \VS_{(n-1)k+1} \leq \VS_{nk+1} \leq \cdots \leq \VS_{\infty k+1} = \V_{k}.
\label{eq:sequence-SDPminimalvalue-rectangular}
\end{align}
The specified level $N=(n-1)k+1$ is extracted for the reason that it admits $k$-row rectangular shape Young diagram $(n^{k})$. We then consider the corresponding SDP based on rectangular shape as defined in Definition~\ref{definition:ReducedSDP-Uk}, namely, $\VS_{(n^{k})}$. There may be sequence of inequalities $\VS_{(1^{k})} \leq \VS_{(2^{k})} \leq \cdots \leq \VS_{(n^{k})} \leq \cdots$, but it is not necessary for the proceeding analysis. What we need is the following inequality:
\begin{align}
\VS_{nk-k+1} \leq \VSw_{(n^{k})} \leq \frac{n+k-1}{k(kn-k+1)} \V_{k}.
\end{align}
\begin{proposition} \label{pro:SDPs-rectangularYoung}
The following relation between the minimal values of the SDPs defined in Definition~\ref{definition:SDP-k-block-positivity-N-BSE}, in Definition~\ref{definition:ReducedSDP-revised} and the optimization problem in Definition~\ref{definition:SDP-k-block-positivity} holds:
\begin{align}
\VS_{kn-k+1}
\leq
\VSw_{(n^{k})}
\leq
\frac{1}{k^2} \V_{k} \leq 0,
\label{eq:Bound-rectangularYoung}
\end{align}
Here, $\VS_{kn-k+1}$ is the minimal value of level $N=kn-k+1$ SDP defined in Definition~\ref{definition:SDP-k-block-positivity-N-BSE}. 
It implies
that rectangular shape scheme, i.e., defining hierarchical SDPs regarding Definition~\ref{definition:ReducedSDP-Uk} underlying the following sequential Young diagrams
\begin{align}
(1^{k}) , (2^{k}) , \ldots , (n^{k}) , \ldots ,
\nonumber
\end{align}
is enough for testing $k$-block-positivity. 
Concretely, let $X$ be the operator whose $k$-block-positivity one wants to test. Then,
\begin{itemize}
\item If there exists a $n \in \BN$ such that $\VS_{(n^{k})} \geq 0$, then it implies $\V_{k}=0$, thus $X$ is certified as $k$-block-positive.
\item If $\VS_{(n^{k})}<0$ for any $n \in \BN$, then it implies $\V_{k}<0$ thus $X$ is not $k$-block-positive.
\item If $\VS_{(n^{k})} \to 0$ as $n \to \infty$, the it implies $\V_{k}=0$ thus $X$ is $k$-block-positive.
\end{itemize}
\end{proposition}
\begin{proof}
Since $\ket{\varphi_{\lambda^{-}}(x,y)}$ in Corollary~\ref{cor:mu-PermInvVec} is a particular form belonging to the feasible set of the SDP in Definition~\ref{definition:ReducedSDP-revised}, we have $\VSw_{\lambda} \leq f_{\lambda}(x,y)$ for any $x \in \BM_{d_{A} \times k}(\BC)$ and $y \in \BM_{d_{B} \times k}(\BC)$.
Let us focus on the rectangular shape scheme and consider the level $n$ of Eq.\eqref{eq:sequence-SDPminimalvalue-rectangular} corresponding to $\lambda=(n^{k})$, then we get $\lambda^{-}=((n-1)^{k})$ and $\mu=(n,(n-1)^{k-1})$ which is the unique $\mu \searrow \lambda^{-}$ under the requirement $\ell(\mu) \leq k$. Revisiting Eq.\eqref{eq:ObjectiveFunction-SchurFunction}, we use Pieri’s formula to deal with
\begin{align}
s_{\lambda^{-}} (\eta_1 , \ldots , \eta_k) \tr (y^{\dagger} y)
=
s_{\lambda^{-}} (\eta_1 , \ldots , \eta_k) p_{1}(\eta_1 , \ldots , \eta_k)=\sum_{\nu : \nu \searrow \lambda^{-}} s_{\nu}
=
s_{\mu}(\eta_1 , \ldots , \eta_k),
\end{align}
where $p_{1}(\eta_1,\ldots,\eta_k)=\sum_{i=1}^{k} \eta_i=\tr \eta$, and $\nu$ is either $\nu=\mu$ or $\nu=((n-1)^{k},1)$ which vanishes since it has rows more than $k$. We obtain
\begin{align}
&
f_{(n^k) , (n,(n-1)^{k-1})}(x,y)
=
\frac{\scp{\varphi_{(n,(n-1)^{k-1})}}{X_{(n^{k})}}{\varphi_{(n,(n-1)^{k-1})}}
}{\inner{\varphi_{(n,(n-1)^{k-1})}}{\varphi_{(n,(n-1)^{k-1})}}}
=
\frac{\dim \BY_{\lambda^{-}}}{\dim \BY_{\mu}} f(x,y),
\label{eq:ObjectiveFunction-SchurFunction-rec}
\end{align}
where the dimensions $\dim \BY_{\lambda^{-}}$ and $\dim \BY_{\mu}$ are computed by hook length formula,
\begin{align}
&
\dim \BY_{((n-1)^k)}=(kn-k)! \prod_{r=1}^{k} \frac{(k-r)!}{(n+k-r-1)!}, \\
&
\dim \BY_{(n,(n-1)^{k-1})}=\frac{k(kn-k+1)!}{n+k-1} \prod_{r=1}^{k} \frac{(k-r)!}{(n+k-r-1)!}.
\end{align}
Suppose that $x_{m} \in \BM_{d_{A} \times k}(\BC) , y_{m} \in \BM_{d_{B} \times k}(\BC)$ minimize $f(x,y)$, i.e., $f(x_{m},y_{m})=\V_{k}$ the solution for the optimization problem in Definition~\ref{definition:SDP-k-block-positivity}, then we have
\begin{align}
\VSw_{(n^{k})}
\leq
f_{(n^k) , (n,(n-1)^{k-1})}(x_{m},y_{m})
=
\frac{n+k-1}{k(kn-k+1)} \V_{k}.
\end{align}
Now we consider two situations separately:
\begin{itemize}
\item If $\VS_{(n^{k})}>0$ then $\VSw_{(n^{k})}=0$, thus above inequality together with $\V_{k} \leq 0$ implies $\V_{k}=0$. Meanwhile, due to sequence Eq.\eqref{eq:sequence-SDPminimalvalue}, $\VS_{kn-k+1} \leq 0$ holds. For this situation, we have
\begin{align}
\VS_{kn-k+1} \leq \VSw_{(n^{k})}=\V_{k}.
\end{align}
\item If $\VS_{(n^{k})} \leq 0$ then $\VS_{(n^{k})}=\VSw_{(n^{k})}$, meanwhile $\VS_{kn-k+1} \leq \VS_{(n^{k})}$ as Eq.\eqref{eq:VSN-minVSlambda}. Hence we have
\begin{align}
\VS_{kn-k+1} \leq \VS_{(n^{k})} = \VSw_{(n^{k})} = \frac{n+k-1}{k(kn-k+1)} \V_{k}.
\end{align}
\end{itemize}
Both situations satisfy the following inequality:
\begin{align}
\VS_{kn-k+1} \leq \VSw_{(n^{k})} \leq \frac{n+k-1}{k(kn-k+1)} \V_{k} \leq \frac{1}{k^2} \V_{k} \leq 0.
\label{eq:Bound-rectangularYoung}
\end{align}
From the inequality one can deduce: (i) if $\VS_{(n^k)} \geq 0$ then $\VSw_{(n^k)}=0$ thus $\V_{k}=0$; (ii) if $\VS_{(n^k)} < 0$ then $\VSw_{(n^k)}=\VS_{(n^k)}$ thus at infinite level we have $\VS_{\infty}=\V_{k} \leq \VS_{(\infty^{k})} \leq \frac{1}{k^2} \V_{k}$, which shows that $\VS_{(\infty^{k})}$ has the same sign with $\V_{k}$. More explicitly, $\VS_{(n^{k})} \to c $ indicates no $k$-block-positivity if $c < 0$, while it indicates $k$-block-positivity if $c=0$.
\end{proof}
\begin{remark}
One may consider the situation based on other $\lambda$. Assume in Eq.\eqref{eq:ObjectiveFunction-SchurFunction} $\dim \BY_{\lambda^{-}} / \dim \BY_{\mu} \to 0$ as $N \to \infty$, then for the situation $\VS_{\lambda} < 0$, the decision on whether $\V_{k}<0$ or $\V_{k}=0$ remains indeterminate, because imitating above analysis one would only get $\VS_{\infty} \leq \VS_{\lambda} \leq 0$, no $\V_{k}$ from the right side, then $\V_{k}<0$ is still possible from $\VS_{\lambda} = 0$.
\end{remark}

\medskip

This section has shown why establishing hierarchical SDPs on rectangular shape is enough for testing $k$-block-positivity. It follows from the fact that the minimal value of the objective function of the reduced SDP in Definition~\ref{definition:ReducedSDP-revised} under rectangular shape is bounded from bottom by the minimal value of the objective function of SDPs defined in Definition~\ref{definition:SDP-k-block-positivity-N-BSE} at the level $N=kn-k+1$ and from the top by the minimal value of the optimization problem in Definition~\ref{definition:SDP-k-block-positivity} up to a nonzero factor.




\section{Complexity of the hierarchical SDPs on rectangular scheme} \label{sec:ComplexitySDP}




\subsection{Complexity of reduced SDP indexed by Young diagram} \label{subsec:Complexity-YoungDiagram}

This subsection aims to determine the size of the SDP variable – the block matrix $\rho_{\lambda}$ in Eq.\eqref{eq:Uk-invariant-rho}) for the SDP defined in Definition~\ref{definition:ReducedSDP-Uk}, by fully utilizing the Bosonic permutation symmetry. The size serves as a measure of the SDP’s complexity and is shown to be related to the dimension of the restricted representation of $\U(d)$.

Consider $\lambda \vdash_{k} (N+k-1)$. In Eq.\eqref{eq:Uk-invariant-rho} we can write $\rho_{p_{\lambda},p'_{\lambda}}=R_{\lambda}(\ketbra{p_{\lambda}}{p'_{\lambda}})$ with a
linear map $R_{\lambda} : \mathbb{M}_{\dim \mathbb{Y}_{\lambda}}(\mathbb{C}) \to \mathbb{M}_{d_{A}d_{B}^{N}}(\mathbb{C})$. 
Then by $\rho_{\lambda} \geq 0$,  $\rho_{\lambda}$ can be represented as
\begin{align}
\rho_{\lambda}
=
\frac{\id_{\BU^{k}_{\lambda}}}{\dim \BU^{k}_{\lambda}} \otimes \sum_{p_{\lambda} , p'_{\lambda}}
\cE \cE^{\dagger}
\ketbra{p_{\lambda}}{p'_{\lambda}}
\cE \cE^{\dagger}
\otimes \sum_{\alpha} K_{\lambda, \alpha} \ketbra{p_{\lambda}}{p'_{\lambda}} K_{\lambda, \alpha}^{\dagger},
\label{eq:rhoYoungDiagram-KrausOperators}
\end{align}
where the Kraus operators $K_{\lambda, \alpha} : \BC^{\dim Y_{\lambda}} \to \BC^{d_{A}} \otimes (\BC^{d_{B}})^{\otimes N}$ should be exchangeable, i.e., satisfy $K_{\lambda, \alpha}=\pi K_{\lambda, \alpha} \pi^{-1}$ for all $\pi \in S_{N}$ since $\rho_{\lambda}$ carries Bosonic permutation symmetry.
Hence, the exchangeable Kraus operators admit the following decomposition,
\begin{align}
K_{\lambda, \alpha} \in
\mathrm{Hom}_{S_{N}} ( \BY_{\lambda} , \BC^{d_{A}} \otimes (\BC^{d_{B}})^{\otimes N} )
&
\cong
\BC^{d_{A}}
\otimes
\bigoplus_{\mu \subset \lambda : \mu \searrow \lambda^{-}} \mathrm{Hom}_{S_{N}} ( \BY_{\mu} , (\BC^{d_{B}})^{\otimes N} )
\\
&
\cong
\BC^{d_{A}} \otimes \bigoplus_{\mu \subset \lambda : \mu \searrow \lambda^{-}} \BU^{d}_{\mu},
\end{align}
where the last step follows by Schur-Weyl duality.

\medskip

Every $\rho_{\lambda}$ can be generated by $\{ K_{\lambda, \alpha} : \BC^{\dim Y_{\lambda}} \to \BC^{d_{A}} \otimes (\BC^{d_{B}})^{\otimes N} \}$ through convex linear combination as Eq.\eqref{eq:rhoYoungDiagram-KrausOperators},
where the vector $\sum_{p_{\lambda}} \ket{p_{\lambda}} \otimes \ket{p_{\lambda}}$ servers as the initial state on which Kraus operators act.
Hence the dimension of $\mathrm{Hom}_{S_{N}} ( \BY_{\lambda} , \BC^{d_{A}} \otimes (\BC^{d_{B}})^{\otimes N} )$ provides a characterization for the size of SDP variables thus the complexity of the computational resource used in the SDP defined in Definition~\ref{definition:ReducedSDP-Uk}.
\begin{proposition} \label{proposition:SDPComplexity-lambda}
The complexity of the SDP in Definition~\ref{definition:ReducedSDP-Uk} is characterized by
\begin{align}
\cC_{\lambda}
:=
\dim \mathrm{Hom}_{S_{N}} ( \BY_{\lambda} , \BC^{d_{A}} \otimes (\BC^{d_{B}})^{\otimes N} )
=
d \sum_{\mu \subset \lambda : \mu \searrow \lambda^{-}} \dim \BU^{d}_{\mu}.
\end{align}
\end{proposition}
The fact that complexity depends on the number of linear independent Kraus operators, may be legitimated as follows: (i) For free constraint case, the SDP variables consists of positive semidefinite $D \times D$ matrices, and the number of convex linear independent SDP variables, is $D^2$; (ii) The set of positive semidefinite $D \times D$ matrices admits a matrix basis (non-orthogonal under Hilbert-Schmidt inner product) with rank-$1$ positive semidefinite $D \times D$ matrices, through $\mathbb{R}$-linear span, i.e.,
\begin{align*}
\BM_{D \times D}(\BC)=\mr{Span}_{\mathbb{R}} \{ h_{i} \geq 0 | i = 1 , \ldots , D^2 \}.
\end{align*}
(iii) These basis matrices can be generated by Kraus operators, e.g., if $h_{i}=\ketbra{v_{i}}{v_{i}}$ then $h_{i}=\ketbra{v_{i}}{1} \cdot \ketbra{1}{1} \cdot \ketbra{1}{v_{i}}$. The number of independent Kraus operators is $D^2$ over $\mathbb{R}$, or $D$ over $\BC$.

\medskip

Let us describe how $\mu \vdash N$ is obtained from $\lambda \vdash (N+k-1)$ via the Littlewood-Richardson rule. The $\lambda$ should be one of Young diagrams obtained from appending $(1^{k-1})$ to some $\mu$, probably expressed as $(1^{k-1}) \otimes \mu \mapsto \lambda$. The allowable $\mu$ may be not unique. Let $\mu^{(i)}$ be one of allowable Young diagrams where $1 \leq i \leq k$. Then we characterize all these $\mu^{(i)}$:
\begin{align}
&
\text{for } i=1:
\quad
\mu^{(1)}=(\lambda_1, \lambda_2-1 , \ldots , \lambda_k-1),
\nonumber \\
&
\text{for } 1<i<k:
\quad
\mu^{(i)}=(\lambda_1-1, \lambda_2-1 , \ldots , \lambda_i , \ldots , \lambda_k-1),
\: \text{ if }
\lambda_{i-1} - 1 \geq \lambda_i \geq \lambda_{i+1} -1,
\nonumber \\
&
\text{for } i=k:
\quad
\mu^{(k)}=(\lambda_1-1, \lambda_2-1 , \ldots , \lambda_k),
\: \text{ if }
\lambda_{k-1} -1 \geq \lambda_{k}.
\label{eq:Procedure1st-lambda-mu}
\end{align}
The admissible condition $\lambda_{i-1} - 1 \geq \lambda_i \geq \lambda_{i+1} -1$ ensures that $\mu^{(i)}$ is a normal shape Young diagram.
For a given $\lambda$ with $\ell(\lambda) = k$, there exists at most $k$ allowable $\mu$.
Alternatively, $\mu$ can be obtained from $\lambda$ by the following equivalent procedure:
\begin{align}
\mu^{(i)}=\underbrace{(\lambda_1-1, \lambda_2-1 , \ldots , \lambda_i-1 , \ldots , \lambda_k-1)}_{=\lambda^{-}} + (0 , 0 , \ldots , \underset{\underset{\text{$i$-th slot}}{\uparrow}}{1} , \ldots , 0).
\label{eq:Procedure2nd-lambda-mu}
\end{align}
That is, $\mu^{(i)}$ is obtained from $\lambda^{-}$ by appending exactly one box to the admissible position meanwhile $\mu^{(i)}$ satisfies $\mu^{(i)} \subset \lambda$.

We illustrate above procedures by the example of $N=4, k=3$ on given $\lambda=(3,2,1)$ in Appendix \ref{sec:Example-recoupling}.

\medskip

At the end of this part, we present a more explicit form of Bosonic permutation symmetric state $\rho_{\lambda}$. We employ Schur basis \cite{bacon2005Schur,chen2025srkbp} to express the formulation. The space of Bob's copies (after $k$-purification) is $( \BC^{kd} )^{\otimes N} \cong ( \BC^{k} )^{\otimes N} \otimes ( \BC^{d} )^{\otimes N}$. We have dealt with $( \BC^{k} )^{\otimes N}$ part under $\U(k)$-symmetry by employing Schur basis, now we turn to $( \BC^{d} )^{\otimes N}$ part via similar manner $( \BC^{d} )^{\otimes N}=\bigoplus_{\mu \vdash N} \BY_{\mu} \otimes \BU_{\mu}$. We then adopt $p_{\mu}$ to label the standard Young tableaux of $\mu$. Likewise, adopt $q_{\mu}$ to label the semi-standard Young tableaux of $\mu$. That is to say, $p_{\mu}$ and $q_{\mu}$ are assigned to label the elements of $\BY_{\mu}$ and of $\BU_{\mu}$, respectively. The Bosonic permutation symmetry admits the following reduction.
\begin{proposition}[Bosonic permutation symmetric state] \label{proposition:main-form-PermInvStates}
The $\rho_{\lambda}$ defined in Eq.\eqref{eq:Uk-invariant-rho} for the SDP in Definition~\ref{definition:ReducedSDP-Uk}, can has the following form:
\begin{align}
&
\rho_{\lambda}
=
\frac{\id_{\mathbb{U}^{\lambda}_{k}}}{\dim \mathbb{U}^{\lambda}_{k}}\otimes
\sum_{\alpha}
\sum_{\mu, \nu  \searrow \lambda^{-}}
\ketbra{\Phi_{\mu}}{\Phi_{\nu}} \otimes \ketbra{\Psi_{\mu}^{\alpha} }{\Psi_{\nu}^{\alpha} },
\\
&
\text{where }
\sum_{\alpha}\sum_{\mu \searrow \lambda^{-}} \inner{\Psi^{\alpha}_{\mu}}{\Psi^{\alpha}_{\mu}}=1,
\: \text{ and }
\ket{\Phi_{\mu}}=\sum_{p_{\mu}=1}^{\dim \mathbb{Y}_{\mu}} \sqrt{ \frac{1}{\dim \mathbb{Y}_{\mu}} } \ket{p_{\mu}} \otimes \ket{p_{\mu}},
\end{align}
where $\mu \vdash N$ and $\ket{\Psi^{\alpha}_{\mu}} \in \mathbb{C}^{d} \otimes \mathbb{U}^{d}_{\mu}$.
\end{proposition}
\begin{proof}
Firstly, $\rho_{\lambda}$ admits spectral decomposition $\rho_{\lambda}=\sum_{\alpha} \frac{\id_{\mathbb{U}^{\lambda}_{k}}}{\dim \mathbb{U}^{\lambda}_{k}}\otimes \ketbra{\psi_{\alpha}}{\psi_{\alpha}}$ with
\begin{align}
\ket{\psi_{\alpha}}=
\sum_{ p_{\lambda}=1 }^{\dim \BY_{\lambda}}
\sum_{\mu \vdash N}
\sum_{ p'_{\mu}=1 }^{\dim \BY_{\mu}}
\sum_{ q_{\mu}=1}^{\dim \BU^{d}_{\mu}} \sum_{i=1}^{d}
\psi^{\alpha}_{p_{\lambda} , i p'_{\mu} q_{\mu}} \ket{p_{\lambda}} \otimes (\ket{i} \otimes \ket{ p'_{\mu} , q_{\mu} })
\end{align}
where $\ket{ p'_{\mu} , q_{\mu} }$ is the Schur basis for $(\mathbb{C}^{d})^{\otimes N}$ with respect to Young diagram $\mu \vdash N$, and $\{ \ket{i} \}_{i=1}^{d}$ the basis for Alice's $\mathbb{C}^{d}$.
We construct Bosonic permutation symmetric vector by averaging,
\begin{align}
\cT_{S_{N}}(\ket{\psi_{\alpha}})=
\frac{1}{N!} \sum_{\pi \in S_{N}} \Delta(\pi) \ket{\psi_{\alpha}},
\end{align}
the statement is then proved by using orthogonality
\begin{align}
\frac{1}{N!}\sum_{\pi \in S_{N}} \pi_{ r_{\nu} p_{\nu} } \pi_{ p''_{\mu} p'_{\mu} }=\frac{1}{\dim \BY_{\mu}} \delta_{\nu \mu} \delta_{r_{\mu} p''_{\mu}} \delta_{p_{\mu} , p'_{\mu}}
\end{align}
where $\pi_{ p''_{\mu} p'_{\mu} } : S_{N} \to \BC$ is the matrix element of the irreducible representation of $S_{N}$ corresponding to $\mu$ and it is always possible to choose all $\pi_{ p''_{\mu} p'_{\mu} }$ to be real numbers.
\end{proof}




\subsection{Complexity underlying rectangular shape of diagrams}

Now we restrict ourselves to the rectangular shape scheme.
Proposition~\ref{proposition:SDPComplexity-lambda} indicates the characterization for the SDP complexity via the dimension of permutation-invariant Kraus operators. Having shown that the hierarchical SDPs for testing $k$-block-positivity can be established on rectangular shape Young diagrams, we will present the explicit formula for the SDP complexity.
Furthermore, this formula explains why extendibility hierarchy would collapse in the $k=d$ case. 

Let us restate the main theorem in Introduction.
\begin{theorem}[The SDP complexity based on rectangular scheme] \label{thm:kBP-SDPComplexity}
Setting $N$ by $N+k-1=k n$ with integer $n$. The SDP complexity of the reduced SDP Definition~\ref{definition:ReducedSDP-Uk} is
\begin{align}
\cC_{(n^{k})}
=d \frac{k(d+n-1)}{k+n-1} \prod_{r=1}^{k} \frac{(d+n-r-1)!(k-r)!}{(k+n-r-1)!(d-r)!}.
\label{eq:thm:SDPComplexity}
\end{align}
\end{theorem}
\begin{proof}
As mentioned earlier, $(n^{k})$ admits unique $\mu$ with $\ell(\mu) \leq k$, that is $\mu=(n, (n-1)^{k-1})$.
The calculation is then straightforward by using formula
\begin{align}
\dim \mathbb{U}_{\mu}^{k}=\prod_{(i,j) \in \mu} \frac{k+j-i}{h_{\mu}(i,j) }.
\label{eq:Dim-unitary-irrep}
\end{align}
to get the wanted dimension as follows,
\begin{align}
\dim \mathbb{U}^{d}_{(n,(n-1)^{k-1})}
=\frac{k(d+n-1)}{k+n-1} \prod_{r=1}^{k} \frac{(d+n-r-1)!(k-r)!}{(k+n-r-1)!(d-r)!}.
\end{align}
Combing this expression with Proposition~\ref{proposition:SDPComplexity-lambda}, the proof is complete.
\end{proof}

\medskip

Two corollaries can be directly read from Eq.\eqref{eq:thm:SDPComplexity} in Proposition~\ref{thm:kBP-SDPComplexity}.
\begin{corollary}[Hierarchy collapse] \label{corollary:SDPHierarchyCollapse}
The SDP complexity is independent on $n$ if $k=d$.
In this situation, we say the SDP hierarchy of $d$-block-positivity testing collapses.
\end{corollary}
This corollary explains that the extendibility hierarchy collapses in the extreme situation $k=d$. Indeed, as already known, if $k=d$ then $k$-block-positivity testing can be easily done by looking whether $X$'s minimal eigenvalue is negative or not. The extendibility hierarchy is thus not needed.

\begin{corollary}[Relation of complexity for distinct $k$]
Let $1 \leq k' \leq k \leq d$, then
\begin{align}
\lim_{n \to \infty}
\frac{ \cC_{(n^{k})} }{ \cC_{(n^{k'})} }
=
\left\{
\begin{aligned}
&
\frac{k}{d-k}, \: &\text{ for } k'=d-k, \\
&
0, \: &\text{ for } k' > d - k, \\
&
\infty, \: &\text{ for } k' < d - k.
\end{aligned} \right.
\end{align}
\end{corollary}
\begin{proof}
The expression
\begin{align}
\frac{ \cC_{(n^{k})} }{ \cC_{(n^{k'})} }
&=
\frac{k}{k'}\frac{k'+n-1}{k+n-1}
\frac{\prod_{r=1}^{k} (\frac{1}{d-r}+\frac{1}{n-1}) \times \cdots \times (\frac{1}{k-r+1}+\frac{1}{n-1}) (n-1)^{d-k}}
{\prod_{m=1}^{k'} (\frac{1}{d-r}+\frac{1}{n-1}) \times \cdots \times (\frac{1}{k'-m+1}+\frac{1}{n-1}) (n-1)^{d-k'}}
\\
&=
\frac{(n-1)^{k(d-k)}}{(n-1)^{k'(d-k')}}
\frac{k}{k'}\frac{k'+n-1}{k+n-1}
\frac{\prod_{r=1}^{k} (\frac{1}{d-r}+\frac{1}{n-1}) \times \cdots \times (\frac{1}{k-r+1}+\frac{1}{n-1})}
{\prod_{m=1}^{k'} (\frac{1}{d-r}+\frac{1}{n-1}) \times \cdots \times (\frac{1}{k'-m+1}+\frac{1}{n-1})}
\\
&=
(n-1)^{(k-k')(d-k-k')}
\frac{k}{k'}\frac{k'+n-1}{k+n-1}
\frac{\prod_{r=1}^{k} (\frac{1}{d-r}+\frac{1}{n-1}) \times \cdots \times (\frac{1}{k-r+1}+\frac{1}{n-1})}
{\prod_{m=1}^{k'} (\frac{1}{d-r}+\frac{1}{n-1}) \times \cdots \times (\frac{1}{k'-m+1}+\frac{1}{n-1})}.
\end{align}
As $n \to \infty$, the expression diverges unless $(k-k')(d-k-k') \leq 0$: if $d-k-k'<0$ then the limit vanishes while if $d-k-k'=0$ then the limit converges to $k/k'$.
\end{proof}
This corollary implies that testing $k$-block-positivity is as hard as testing $k'$-block-positivity in the sense of big O notation.

\paragraph{Pictorial insight}
The phenomenon of hierarchy collapse can be also pictorially read from a lemma of representation theory.
\begin{lemma} \label{lemma:DimUd-RectangleShape-Skew}
Suppose $\mu \subseteq (n^{d})$, then $\dim \BU^{d}_{\mu}=\dim \BU^{d}_{(n^{d}) / \mu}$.
\end{lemma}
For rectangular Young diagram $(n^{k})$, Lemma~\ref{lemma:DimUd-RectangleShape-Skew} shows $\dim \BU^{d}_{(n,(n-1)^{k-1})}=\dim \BU^{d}_{(n^{d-k},1^{k-1})}$.
For instance, if $d=5, n=4$, then we add the Young diagram into a $5$-row rectangle,
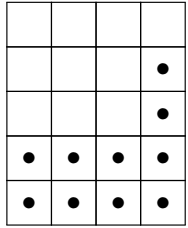
\begin{align*}
  \begin{ytableau}
   \quad & \quad & \quad & \quad \\
   \quad & \quad & \quad & \bullet \\
   \quad & \quad & \quad & \bullet \\
   \bullet & \bullet & \bullet & \bullet \\
   \bullet & \bullet & \bullet & \bullet
  \end{ytableau}
\end{align*}
Lemma~\ref{lemma:DimUd-RectangleShape-Skew} tells that $\dim \BU^{5}_{(4,3^2)}=\dim \BU^{5}_{(4^2,1^2)}$.
Then Corollary~\ref{corollary:SDPHierarchyCollapse} has a pictorial explanation thanks for Lemma~\ref{lemma:DimUd-RectangleShape-Skew}: In the case $k=d$, its complementary is $(1^{d-1})$, showing that the corresponding dimension of irreducible $\U(d)$ is $d$.




\section{Conclusion}

In this paper, we have derived the SDP complexity of a semidefinite programming hierarchy for testing $k$-block-positivity, building upon the framework established in \cite{chen2025srkbp}. We showed that it is sufficient to implement the hierarchical SDPs using only rectangular Young diagrams. This approach opens the possibility of solving reduced SDPs involving a smaller set of irreducible representations, rather than running the full decomposition. Furthermore, we formulated permutational invariance in terms of exchangeable Kraus operators, instead of incorporating it as additional SDP constraints as in \cite{chen2025srkbp}. This allows the computational complexity to be quantified directly in terms of the dimensions of the irreducible representations of $\U(d)$, where $d$ denotes the dimension of the original local space prior to $k$-extension. Finally, we anticipate that these methodological results will provide useful insights into computational approaches for challenging optimization problems in quantum information theory, such as the 2-copy distillability conjecture.




\section*{Acknowledgements}

We are grateful to Omar Fawzi for valuable insights during the initial stages of the project, as well as to both Omar Fawzi and Kun Fang for their helpful feedback and discussions.
Q. C. is supported by grants NSFC 12341101, and the QuantERA II Programme that has received funding from the European Union’s Horizon 2020 research and innovation programme under Grant Agreement No 101017733 (VERIqTAS).
B. C. is supported by JSPS Grant-in-Aid Scientific Research (A) no.
25H00593, and Challenging Research (Exploratory) no. 23K17299.
B. C. was also supported by the Chaire Jean Morlet and acknowledges the hospitality of ENS Lyon during the fall 2024, which allowed us to initiate this project. 




\begin{appendices}


\section{The proof of Proposition~\ref{proposition:Bound-SNK-SN1-SNk}} \label{proof:proposition:Bound-SNK-SN1-SNk}

We present a proof for Proposition~\ref{proposition:Bound-SNK-SN1-SNk}.
\begin{proof}
Define map $\hat{\phi}_{k} : \BC^{k d_{A}} \otimes \BC^{k d_{B}} \to \BC^{d_{A}} \otimes \BC^{d_{B}}$ by setting $\hat{\phi}_{k}=\ket{\phi_{k}} \otimes \id_{d_{A}} \otimes \id_{d_{B}}$ that satisfies $\hat{\phi}_{k}^{\dagger} \hat{\phi}_{k}=k \id_{d_{A}} \otimes \id_{d_{B}}$. Then for any $\ket{\varphi} \in \mr{SR}_{1} (kd_{A} , kd_{B})$, we have $0 \leq \|{\hat{\phi}_{k}^{\dagger} \ket{\varphi}}\|_{2} \leq 1$. This can be shown by noting that every $\ket{\varphi} \in  \mr{SR}_{1} (kd_{A} , kd_{B})$ can be written as $\ket{\varphi}=\sum_{i \in I_{A}} \sum_{j \in J_{B}} \ket{i j } \otimes (x \otimes y) \ket{i , j}$ where $I_{A}$ and $J_{B}$ are index sets; $x : \BC^{k} \to \BC^{d_{A}}$ and $y : \BC^{k} \to \BC^{d_{B}}$ satisfy $\tr (x^{\dagger} x) \tr (y^{\dagger} y)=1$ due to the normalization $\inner{\varphi}{\varphi}=1$. Then since $\|{\hat{\phi}_{k}^{\dagger} \ket{\varphi}}\|_{2}^2=\tr(x^{\dagger} x y^{\tp} \bar{y})$ has the form of $\tr (A B^{\dagger})$, so $0 \leq \|{\hat{\phi}_{k}^{\dagger} \ket{\varphi}}\|_{2}^2 \leq 1$ where if $I_{A} \cap J_{B}=\emptyset$ then $\hat{\phi}_{k}^{\dagger} \ket{\varphi}=0$; if $I_{A}=J_{B}$ are singletons, then $\| \hat{\phi}_{k}^{\dagger} \ket{\varphi} \|_{2}=1$.

Set $\mr{SR}_{1} (kd_{A} , kd_{B})$ to contain either $\hat{\phi}_{k}^{\dagger} \ket{\varphi} \neq 0$ or $\hat{\phi}_{k}^{\dagger} \ket{\varphi'}=0$.
On the one hand, the inequality $\V_{k} \leq 0$ holds because if such $\ket{\varphi'}$ exists, hence if $\V \geq 0$, then $\V_{k}=0$. On the other hand, every nonzero $\ket{\psi} \in \mr{SR}_{k} (d_{A},d_{B})$ admits Schmidt a decomposition
\begin{align}
\ket{\psi}=\sum_{i=1}^{k} \ket{x_{i} \otimes y_{i}},
\: \text{ with }
\inner{x_{j}}{x_{i}}=\sqrt{s_{i}} \delta_{ji}
\: \text{ and }
\inner{y_{j}}{y_{i}}=\delta_{ji},
\end{align}
where the singular values $s_{i} \geq 0$ and satisfy $\sum_{i=1}^{k} s_{i}=1$, hence it can be purified into
\begin{align}
\mr{SR}_{1} (kd_{A} , kd_{B}) \ni
\ket{\varphi}=\frac{1}{\sqrt{\mr{sr}(\psi)}} \sum_{i,j=1}^{k} \cE \ket{i \otimes j} \otimes \ket{x_{i} \otimes y_{j}},
\end{align}
where $\mr{sr}(\psi)$ is the Schmidt rank of $\ket{\psi}$, the factor $1 / \sqrt{\mr{sr}(\psi)}$ is for normalization $\inner{\varphi}{\varphi}=1$. Let $\ket{\psi_{m}}$ be a minimizer of the optimization problem in Definition~\ref{definition:Optimization-k-block-positivity-original}. Then $\V=\scp{\psi_{m}}{X}{\psi_{m}}$. The above purification shows that $\V_{k} \leq \V / \mr{sr}(\psi_{m})$.
From now on, we focus on $\V<0$. Then, $k \V_{k} \leq \mr{sr}(\psi_{m}) \V_{k} \leq \V < 0$ where $k \geq \mr{sr}(\psi_{m})$. Let $\ket{\varphi''} \in \mr{SR}_{1} (kd_{A} , kd_{B})$ be the vector such that $\V_{k}=\scp{\varphi''}{X_{k}}{\varphi''}$. Then $\|{\hat{\phi}_{k}^{\dagger} \ket{\varphi''}}\|_{2} \neq 0$ is deduced from $\V_{k} < 0$. Using $0 < \|{\hat{\phi}_{k}^{\dagger} \ket{\varphi''}}\|_{2} \leq 1$ we obtain
\begin{align}
0 > \V_{k} = \scp{\varphi''}{k \Pi_{k} \otimes X}{\varphi''}
=
\|{\hat{\phi}_{k}^{\dagger} \ket{\varphi''}}\|_{2}^2 \scp{\psi''}{X}{\psi''}
\geq
\scp{\psi''}{X}{\psi''}
\geq
\V,
\nonumber
\end{align}
where $\mr{SR}_{k} (d_{A},d_{B}) \ni \ket{\psi''}= \hat{\phi}_{k}^{\dagger} \ket{\varphi''} / \|{\hat{\phi}_{k}^{\dagger} \ket{\varphi''}}\|_{2}$. Hence $0 > \V_{k} \geq \V \geq \mr{sr}(\psi_{m}) \V_{k} \geq k \V_{k}$.
\end{proof}




\section{Example of recoupling} \label{sec:Example-recoupling}

This section aims to illustrate the recoupling on given Young diagram, from the case of $k=3$ and $N=4$. To be specific, let us consider by starting with a given Young diagram, say, $\lambda=(3,2,1)$, then we ask which $\mu \vdash 4$ can yield $(3,2,1)$ via tensoring $(1^2)$ that represents $\cE \cE^{\dagger}$.
According to Eq.\eqref{eq:Procedure1st-lambda-mu}, the answer of the question is the following:
\begin{align*}
\begin{ytableau} \bullet & \quad & \quad \\ \bullet & \quad \\ \quad \end{ytableau}
\cong
\begin{ytableau} \quad \\ \quad \end{ytableau}
\otimes
\begin{ytableau} \quad & \quad & \quad \\ \quad \end{ytableau}
\oplus
\begin{ytableau} \quad \\ \quad \end{ytableau}
\otimes
\begin{ytableau} \quad & \quad \\ \quad & \quad \end{ytableau}
\oplus
\begin{ytableau} \quad \\ \quad \end{ytableau}
\otimes
\begin{ytableau} \quad & \quad \\ \quad \\ \quad \end{ytableau}.
\end{align*}
The consistent check on the dimension of symmetric group is done by $8=1 \times (4+2+4)$ where $8$ is the total number of the skew Young tableaux of $(3,2,1) / (1^2)$.

One can answer the question from the manner described in Eq.\eqref{eq:Procedure2nd-lambda-mu}, by appending $(1^{k-1})$ to $\lambda^{-}$. Write
\[
(3,2,1)=\begin{ytableau} \quad & \quad & \quad \\ \quad & \quad \\ \quad \end{ytableau}, \: \ (3,2,1)^{-}=(2,1)=\begin{ytableau} \quad & \quad \\ \quad \end{ytableau},
\]
the based on $\lambda^{-}=(2,1)$, an allowable $\mu$ is obtained by appending one box, which should be one of
\[
\begin{ytableau} \quad & \quad \\ \quad \\ \quad \end{ytableau}, \
\begin{ytableau} \quad & \quad & \quad \\ \quad \end{ytableau}, \
\begin{ytableau} \quad & \quad \\ \quad & \quad \end{ytableau}.
\]
Now one is able to check the tensor $\cE \cE^{\dagger}$ with $\mu$ where $\cE \cE^{\dagger}$ is represented by Young diagram $(1^2)$. This is done by using Littlewood-Richardson rule for decomposing $(1^2) \otimes \mu$ as follows,
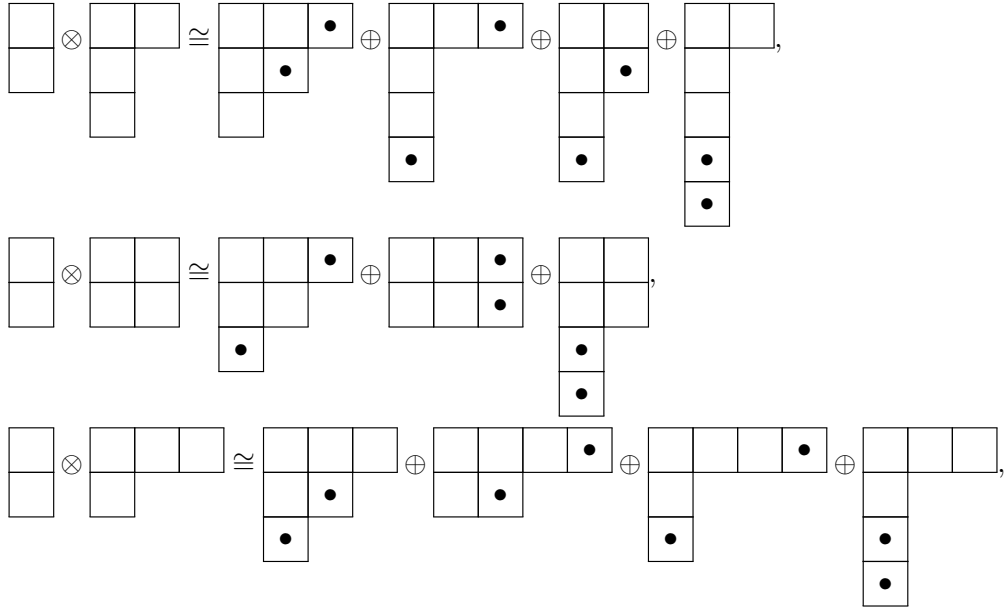
\begin{align*}
&
\begin{ytableau} \quad \\ \quad \end{ytableau}
\otimes
\begin{ytableau} \quad & \quad \\ \quad \\ \quad \end{ytableau}
\cong
\begin{ytableau} \quad & \quad & \bullet \\ \quad & \bullet \\ \quad \end{ytableau}
\oplus
\begin{ytableau} \quad & \quad & \bullet \\ \quad \\ \quad \\ \bullet  \end{ytableau}
\oplus
\begin{ytableau} \quad & \quad \\ \quad & \bullet \\ \quad \\ \bullet \end{ytableau}
\oplus
\begin{ytableau} \quad & \quad \\ \quad \\ \quad \\ \bullet \\ \bullet \end{ytableau},
\\
&
\begin{ytableau} \quad \\ \quad \end{ytableau}
\otimes
\begin{ytableau} \quad & \quad \\ \quad & \quad \end{ytableau}
\cong
\begin{ytableau} \quad & \quad & \bullet \\ \quad & \quad \\ \bullet \end{ytableau}
\oplus
\begin{ytableau} \quad & \quad & \bullet \\ \quad & \quad & \bullet \end{ytableau}
\oplus
\begin{ytableau} \quad & \quad \\ \quad & \quad \\ \bullet \\ \bullet \end{ytableau},
\\
&
\begin{ytableau} \quad \\ \quad \end{ytableau}
\otimes
\begin{ytableau} \quad & \quad & \quad \\ \quad \end{ytableau}
\cong
\begin{ytableau} \quad & \quad & \quad \\ \quad & \bullet \\ \bullet \end{ytableau}
\oplus
\begin{ytableau} \quad & \quad & \quad & \bullet \\ \quad & \bullet \end{ytableau}
\oplus
\begin{ytableau} \quad & \quad & \quad & \bullet \\ \quad \\ \bullet \end{ytableau}
\oplus
\begin{ytableau} \quad & \quad & \quad \\ \quad \\ \bullet \\ \bullet \end{ytableau},
\end{align*}
where $\bullet$ stands for the boxes from $(1^2)$. It is clear that they all produce $(3,2,1)$, and no other $\mu \vdash 4$ produces $(3,2,1)$.

\end{appendices}







\bibliographystyle{alpha_abbrv}
\bibliography{references}

\end{document}